\documentclass[letterpaper, 10 pt, conference]{ieeeconf}  

\IEEEoverridecommandlockouts                              

\usepackage{amsmath,amsfonts}

\usepackage[colorlinks,urlcolor=blue,linkcolor=blue,citecolor=blue]{hyperref}
\usepackage{cite}

\usepackage{caption}
\usepackage{subfigure}

\usepackage{amssymb}
\usepackage{amsthm}
\usepackage{dsfont}
\usepackage{mathtools}
\usepackage{threeparttable}
\usepackage{colortbl}
\usepackage{siunitx}
\usepackage{multirow}
\usepackage{ifthen} 
\newboolean{showcomments}
\setboolean{showcomments}{true}   
\usepackage{todonotes}

\definecolor{bleudefrance}{rgb}{0.19, 0.55, 0.91}
\definecolor{ao(english)}{rgb}{0.0, 0.5, 0.0}

\newcommand{\enrique}[1]{  \ifthenelse{\boolean{showcomments}}
{\todo[inline,color=bleudefrance]{Enrique says: #1}}{}}

\newcommand{\addcite}[0]{\ifthenelse{\boolean{showcomments}}
{\textcolor{purple}{(add cite(s)) }}{}}%

\newcommand{\emmargin}[1]{\ifthenelse{\boolean{showcomments}}
{\todo{Enrique: #1)}}{}}

\newboolean{showedits}
\setboolean{showedits}{true}
\usepackage[markup=underlined]{changes}
\definechangesauthor[color=bleudefrance]{EM}
\newcommand{\adde}[1]{
\ifthenelse{\boolean{showedits}}
{\added[id=EM]{#1}}
{\!#1\hspace{-4.75pt}}
}
\newcommand{\repe}[2]{
\ifthenelse{\boolean{showedits}}
{\replaced[id=EM]{#1}{#2}}
{\!#1\hspace{-4.75pt}}
}
\newcommand{\dele}[1]{
\ifthenelse{\boolean{showedits}}
{\deleted[id=EM]{#1}}
{}
}

\definechangesauthor[color=red]{YJ}
\newcommand{\addy}[1]{
\ifthenelse{\boolean{showedits}}
{\added[id=YJ]{#1}}
{\!#1\hspace{-4.75pt}}
}
\newcommand{\repy}[2]{
\ifthenelse{\boolean{showedits}}
{\replaced[id=YJ]{#1}{#2}}
{\!#1\hspace{-4.75pt}}
}
\newcommand{\dely}[1]{
\ifthenelse{\boolean{showedits}}
{\deleted[id=YJ]{#1}}
{}
}

\DeclareSIUnit[]{\pu}{p.u.}
\DeclareSIUnit[]{\VA}{VA}

\DeclareSymbolFont{bbold}{U}{bbold}{m}{n}
\DeclareSymbolFontAlphabet{\mathbbold}{bbold}

\newcommand{\diag}[1]{\ensuremath{\mathrm{diag}(#1)}}

\DeclarePairedDelimiterX\Set[2]{\lbrace}{\rbrace}%
{ #1 \,\delimsize| \,\mathopen{} #2 }

\newcommand{\real}[0]{\mathbb R}

\usepackage{tikz}
\newcommand*\circled[1]{\tikz[baseline=(char.base)]{\node[shape=circle,draw,inner sep=0.05pt] (char) {#1};}}
\newtheoremstyle{bfnote}%
{}{}%
{\itshape}{}%
{\bfseries}{.}%
{ }%
{\thmname{#1}\thmnumber{ #2}\thmnote{ (#3)}}
\theoremstyle{bfnote}
\newtheorem{thm}{Theorem}

\newtheorem{rem}{Remark}
\newtheorem{lem}{Lemma}
\newtheorem{ass}{Assumption}
\newtheorem{definition}{Definition}

\allowdisplaybreaks
\let\labelindent\relax
\usepackage{enumitem}
\setlist[itemize]{leftmargin=*}

\title{\LARGE \bf
Adaptive Pricing for Optimal Coordination in Networked Energy Systems with Nonsmooth Cost Functions
}

\author{Jiayi Li$^{1}$, Jiale Wei$^{2}$, Matthew Motoki$^{1}$, Yan Jiang$^{2}$, and Baosen Zhang$^{1}$
\thanks{$^{1}$J. Li, M. Motoki and B. Zhang are with the Electrical and Computer Engineering, University of Washington, Seattle, WA. They are partially supported NSF grant ECCS-2153937 and the Washington Clean Energy Institute. {\tt\small \{ljy9712,mmotoki,zhanbao\}@uw.edu}}%
\thanks{$^{2}$ J. Wei and Y. Jiang are with the School of Science and Engineering, The Chinese University of Hong Kong, Shenzhen, GD 518172, CHN (email: {\tt\small jialewei@link.cuhk.edu.cn}; {\tt\small yjiang@cuhk.edu.cn}). Y. Jiang is partially supported by CUHKSZ University Development Fund and part of this research was performed while the author was visiting the
Institute for Mathematical and Statistical Innovation at University of Chicago, which is supported by the NSF grant DMS-1929348.} 
}

\begin{document}

\maketitle
\thispagestyle{empty}
\pagestyle{empty}

\begin{abstract}
Incentive-based coordination mechanisms for distributed energy consumption have shown promise in aligning individual user objectives with social welfare, especially under privacy constraints. Our prior work proposed a two-timescale adaptive pricing framework, where users respond to prices by minimizing their local cost, and the system operator iteratively updates the prices based on aggregate user responses. A key assumption was that the system cost need to smoothly depend on the aggregate of the user demands. In this paper, we relax this assumption by considering the more realistic model of where the cost are determined by solving a DCOPF problem with constraints. We present a generalization of the pricing update rule that leverages the generalized gradients of the system cost function, which may be nonsmooth due to the structure of DCOPF. We prove that the resulting dynamic system converges to a unique equilibrium, which solves the social welfare optimization problem. Our theoretical results provide guarantees on convergence and stability using tools from nonsmooth analysis and Lyapunov theory. Numerical simulations on networked energy systems illustrate the effectiveness and robustness of the proposed scheme.
\end{abstract}

\section{INTRODUCTION}
Modern energy systems are undergoing a significant transformation, marked by the increasing prevalence of distributed energy resources (DERs), responsive loads, and the emergence of more autonomous devices. These developments have created opportunities for customers to actively participate in system operations. However, unlike dispatchable resources, customers often cannot be directly controlled by an operator~\footnote{ Direct load control exists and have been implemented, but are often constrained by the number of times they can be called and duration~\cite{ruiz2009direct,chen2014distributed}, and we do not explore this class of resources in this paper.} and must be coordinated through some form of incentives~\cite{rahimi2010demand}. But the system and its customers often have competing objectives: system operators strive to achieve global objectives like efficiency, reliability, fairness and stability of the network, individual users optimize their private costs and preferences that are often unknown or unobservable. In this paper, we study how to achieve alignment between the system objective and user objective while keeping most of the information about the users private. 


Incentive-based coordination mechanisms have received extensive attention and are one of the main features of power systems with communication capabilities. In the context of demand response in electricity markets, incentives can take many different forms, ranging from alert/text-based signals~\cite{peplinski2023residential} to pricing~\cite{vardakas2014survey}. In this paper, we focus on price-based incentives: a system operator broadcasts prices, users respond by adjusting their consumption to minimize their individual costs, the operator adjusts the prices based on the user responses, etc. Ideally, this iterative interaction should converge to an optimal solution that balances user cost and system performance. The major obstacle is that the operator typically lacks access to users' cost functions, either due to privacy concerns or because users themselves rely on complex or black-box control strategies (e.g., reinforcement learning)~\cite{li2017distributed,khezeli2017risk,kong2020online}. This limits the effectiveness of many pricing schemes and makes theoretical analysis difficult.


A previous work \cite{LI2024EPSR} proposed a two-timescale adaptive pricing framework that is adopted from a dynamic incentive \cite{maheshwari2022inducing} that evolves with the actions of the users. 
In this framework, users act as price takers, optimizing their local behavior in response to a broadcast price signal, while the operator iteratively updates prices based on observed aggregate consumption.
Its iterative update circumvents the need for user-specific knowledge. The main result showed that under mild conditions--such as monotonicity of user response with respect to price--this adaptive scheme converges to the solution of a global social welfare optimization problem.


However,~\cite{LI2024EPSR} made a key simplifying assumption: that prices and the operator's objective are a function of aggregate demand alone (hence the prices are uniform across the users).  Of course, in real-world power systems, electricity must be delivered over a physical network, where supply and demand must balance at each node, and transmission line capacities impose additional constraints. On top of these constraints, the operator solves an optimization problem, in this paper modeled as a DCOPF problem, that determines the best way to satisfy the demands. This introduces new layers of complexity, since the cost depends nonlinearly and nonsmoothly on the demand, and the prices can exhibit discontinuities. 

The nonsmoothness of the price arises quite naturally. In DCOPF problems, the feasible regions are polytopic, and when the generator costs are linear, the optimal solutions occur at the vertices of the feasible region~\cite{stott2009dc}. Therefore, a small change in load can change the set of binding constraints and in turn cause discontinuous jumps in the prices~\cite{zhang2021convex}. The algorithms in \cite{LI2024EPSR} and \cite{maheshwari2022inducing} use prices to infer gradient information about the system cost, but when the cost in DCOPF is nonsmooth and the prices are discontinuous, gradients are no longer well-defined.  



This paper extends the adaptive pricing framework by embedding DCOPF constraints into the system operator's objective and carefully designing pricing updates based on  generalized gradients of the (possibly nonsmooth) cost. 
Our formulation integrates network constraints directly into the operator’s cost, and we propose a pricing update rule based on generalized gradients. This rule accounts for the potential non-differentiability of the cost function due to network constraints. Our main contributions are:
\begin{enumerate}
    \item We design a vectorized price update rule based on the generalized gradient of the nonsmooth system cost induced by DCOPF, enabling implementation in realistic grid models.
    \item We prove that the proposed iterative mechanism converges to a unique equilibrium that aligns user behavior with the social welfare solution. The proof handles both linear and quadratic cost structures, using tools from convex analysis and Lyapunov stability theory for nonsmooth systems.
\end{enumerate}
This work offers a scalable and theoretically grounded approach to aligning local and global objectives in networked energy systems, opening the door to practical decentralized control under realistic grid constraints. The proposed mechanism is robust to privacy constraints, as the operator requires only demand observations and users do not need to disclose their cost functions or internal constraints, thus preserving privacy. We also demonstrate through simulations on networked scenarios that the mechanism effectively induces socially optimal behavior while maintaining system feasibility under DCOPF.

\section{Problem Formulation}

\subsection{Planner's Optimization Problem}\label{ssec:plan}
We consider a supply-demand balancing electricty market with $n$ users indexed by $i\in \mathcal{N} :=\{1,\dots, n\} $. The power demand of user $i$ is denoted by $\boldsymbol{x}_i\in\real$.
For a given demand profile of users $\vec{\boldsymbol{x}}:=\left(\boldsymbol{x}_i, i \in \mathcal{N} \right) \in \real^{n}$ which is the column vector obtained from the concatenation of demand vectors of all users, the disutility (or cost) of the power consumption of user $i$ is given by $f_i(\boldsymbol{x}_i)$, while the system cost in serving the demand profile of users is given by $J(\vec{\boldsymbol{x}})$. 


We now discuss them separately:
\begin{ass}[Cost assumption]\label{ass:cost}
The following assumptions on cost functions are made throughout this manuscript:
\begin{itemize}
\item Each user disutility function $f_i(\boldsymbol{x}_i)$ is strictly convex and twice continuously differentiable;
\item The system cost function $J(\vec{\boldsymbol{x}})$ is a parametric programming determined from the DCOPF problem with linear generation costs:
\begin{align}\label{eq:Jcostfun}
		\!\!\!\!\!\!\!\!\!\!\!J(\vec{\boldsymbol{x}})\!:=\!\min_{\boldsymbol{\xi}}&\  \boldsymbol{c}^T\boldsymbol{\xi}   \\
		\mathrm{s.t.}& \  
		\mathrm{Linear\ constraints\ on\ } \boldsymbol{\xi} \mathrm{\ depending\ on\ }  \vec{\boldsymbol{x}}\nonumber
\end{align}

where $\boldsymbol{c}:=\left(c_i, i \in \mathcal{G} \right) \in \real^{|\mathcal{G}|}$ denotes the vector of generation cost coefficients and $\boldsymbol{\xi}:=\left(\xi_i, i \in \mathcal{G} \right) \in \real^{|\mathcal{G}|}$ denotes the power generation from a set of generators $\mathcal{G}$. A nice feature of the optimal cost $J(\vec{\boldsymbol{x}})$ is that it is a convex function of the user demand profile $\vec{\boldsymbol{x}}$~\cite{Bertsimas1997linear}. However, although $J(\vec{\boldsymbol{x}})$ is continuous, it is not differentiable everywhere.
\end{itemize} 
\end{ass}
Then the system operator is interested in solving the following \emph{global social welfare problem}:
\begin{equation}\label{eq:opt-plan}
		\min_{\vec{\boldsymbol{x}}\in \real^{nT}}\quad      C(\vec{\boldsymbol{x}}):=\sum_{i \in \mathcal{N}} f_i(\boldsymbol{x}_i) +J(\vec{\boldsymbol{x}})\,,
\end{equation}
which minimizes the sum of the total disutility of all users and the system cost to serve users. 
\begin{rem} [Linear generation cost assumption]
Note that the linear cost in~\eqref{eq:Jcostfun} is in some sense the most difficult cost function to deal with, at least in our setting. If the cost is strongly convex function, for example, a quadratic cost, $J$ becomes differentiable everywhere. All of the results in the paper still hold since in that case the generalized gradient is the (standard) gradient and all sets are singletons. Therefore, we focus on linear cost functions in this paper.   
\end{rem}

As discussed before, we adopt the standard assumption that each disutility function $f_i(\boldsymbol{x}_i)$ is strictly convex and twice continuously differentiable~\cite{Mallada2017tac, DORFLER2017Auto} while the system cost function $J(\vec{\boldsymbol{x}})$ is convex and locally Lipschitz\footnote{Every convex function is locally Lipschitz~\cite{WSU1972AM}. We list locally Lipschitz property explicitly for the purpose of emphasis.}. Hence, it is easy to see that the entire objective function is strictly convex and locally Lipschitz, which implies the existence of a unique global minimizer $\vec{\boldsymbol{x}}^\star$ to problem \eqref{eq:opt-plan}~\cite[Proposition 3.1.1]{Bertsekas2009Convex}. By~\cite[Theorem 8.2]{Wright2022Optimization}, $\vec{\boldsymbol{x}}^\star$ is such a minimizer if and only if
\begin{align}\label{eq:optimal-planner}
    \boldsymbol{0}&\in \partial\left(\sum_{i \in \mathcal{N}} f_i(\boldsymbol{x}^\star_i) +J(\vec{\boldsymbol{x}}^\star)\right)\nonumber\\&=\left(\nabla f_i(\boldsymbol{x}^\star_i), i \in \mathcal{N} \right)+\partial J(\vec{\boldsymbol{x}}^\star)\,,
\end{align}
where the equality is due to the sum rule of the generalized gradient for convex functions~\cite[Chapter 2.4]{Clarke1998Nonsmooth}. The so-called generalized gradient is a counterpart to gradient for nonsmooth functions, which is often known to be subdifferential by optimization community. As mentioned in Assumption~\ref{ass:cost}, $J(\vec{\boldsymbol{x}})$ is continuous but it is not differentiable everywhere, which forces us to borrow the generalized gradient concept. 
\begin{definition}[Generalized gradient~\cite{Clarke1998Nonsmooth}]\label{def:def-generalized gra}
If $g:\real^d\mapsto \real$ is a locally Lipschitz continuous function, then its generalized gradient $\partial g:\real^d\mapsto \mathcal{B}(\real^d)$ at $\boldsymbol{z}\in\real^d$ is defined by 
\begin{align*}
    \partial g (\boldsymbol{z}):= \mathrm{co}\left\{\lim_{k\to\infty}\nabla g(\boldsymbol{z}_k):\boldsymbol{z}_k\to\boldsymbol{z},\boldsymbol{z}_k\notin\Omega_g\cap\mathcal{S}\right\}\,,
\end{align*}
where $\mathrm{co}$ denotes convex hull, $\Omega_g\subset\real^d$ denotes the set of points where $g$ fails to be differentiable, and $\mathcal{S}\in\real^d$ is a set of measure zero that can be arbitrarily chosen to simply the computation.
\end{definition}
\begin{rem}[Relation to gradient]
    Unlike a gradient which gives a single vector, a generalized gradient is a set-valued map. The generalized gradient is the generalization of the gradient in the sense that, if $g$ is differentiable at $\boldsymbol{z}$, then $\partial g (\boldsymbol{z})=\{\nabla g (\boldsymbol{z})\}$.
\end{rem}

However, in practice, the planner's optimization problem in \eqref{eq:opt-plan} is not implementable due to the lack of knowledge of the exact disutility functions of users for privacy concerns. This poses challenges for the system operator to realize economic dispatch by solving \eqref{eq:opt-plan} directly. An important way to address this issue by the system operator is to update its power price $\boldsymbol{p}_i\in\real$ for individual users iteratively based on how users adjust their desired power. By doing so, the system operator hopes to encourage users to align their individual goals of cost minimization with the goal of problem \eqref{eq:opt-plan}. The design of such an adaptive price update will be discussed later, which is the core of this manuscript. 

\subsection{User's Optimization Problem}
All users are assumed to be rational price takers. More precisely, given the power price $\boldsymbol{p}_i$, each user $i$ adjusts its power consumption by solving the following optimization problem:
\begin{equation}\label{eq:opt-user}
		\min_{\boldsymbol{x}_i}\quad f_i(\boldsymbol{x}_i) +\boldsymbol{p}_i^T\boldsymbol{x}_i\,,
\end{equation}
which minimizes the total cost of user $i$ induced by disutility and payment for power consumption. Since \eqref{eq:opt-user} is an unconstrained convex optimization problem, the necessary and sufficient condition for $\boldsymbol{x}^*_i$ to be a minimizer is \cite[Chapter 4.2.3]{Boyd2004convex}
\begin{align}\label{eq:x*}
    \nabla f_i(\boldsymbol{x}^*_i)+\boldsymbol{p}_i=0\,,
\end{align}
which yields a unique global solution $\boldsymbol{x}^*_i$ by the strict convexity of $f_i(\boldsymbol{x}_i)$~\cite[Proposition 3.1.1]{Bertsekas2009Convex}. 
Basically, as the system operator updates its price signal $\boldsymbol{p}_i$, user $i$ adjusts its power demand $\boldsymbol{x}^*_i$ accordingly to satisfy \eqref{eq:x*} in a unique way. To put it another way, for any given price $\boldsymbol{p}_i$, the demand $\boldsymbol{x}^*_i$ is unique. Hence, $\boldsymbol{x}^*_i$ is clearly a function \cite[Definition 2.1]{Rudin2013PMA} of the current price $\boldsymbol{p}_i$ and can be expressed as 
\begin{align*}
\boldsymbol{x}^*_i(\boldsymbol{p}_i):= \arg \min_{\boldsymbol{x}_i} \quad f_i(\boldsymbol{x}_i) +\boldsymbol{p}_i^T\boldsymbol{x}_i  \,.
\end{align*}
An important feature of this function $\boldsymbol{x}^*_i(\boldsymbol{p}_i)$ is that it is a continuously differentiable and strictly decreasing function, which is highlighted by the following lemma. 
\begin{lem}[Bijective demand update] \label{lem:x-p-bijective}Under Assumption~\ref{ass:cost}, the demand update $\boldsymbol{x}^*_i(\boldsymbol{p}_i)$ is a bijection given by a continuously differentiable and strictly decreasing function
\begin{align}\label{eq:demand-formula}
\boldsymbol{x}^*_i(\boldsymbol{p}_i)=\nabla^{-1} f_i(-\boldsymbol{p}_i)\,,
\end{align}
which naturally has the properties that $\nabla\boldsymbol{x}^*_i(\boldsymbol{p}_i)<0$ and, $\forall \tilde{\boldsymbol{p}}_i, \hat{\boldsymbol{p}}_i \in \real$, if $\tilde{\boldsymbol{p}}_i\neq\hat{\boldsymbol{p}}_i$, then $\boldsymbol{x}^*_i(\tilde{\boldsymbol{p}}_i)\neq\boldsymbol{x}^*_i(\hat{\boldsymbol{p}}_i)$. 
\end{lem}
\begin{proof}
First, by~\cite[Theorem~2.14]{rockafellar1998}, the strict convexity of $f_i(\boldsymbol{x}_i)$ ensures that its gradient $\nabla f_i(\boldsymbol{x}_i)$ is a strictly increasing function. Then, this strict monotonicity implies that $\nabla f_i(\boldsymbol{x}_i)$ is a bijection, which further implies that $\nabla f_i(\boldsymbol{x}_i)$ has a unique inverse function written as $\nabla^{-1} f_i$ that is also a bijection. Now, we note that, as an optimal solution, $\boldsymbol{x}^*_i(\boldsymbol{p}_i)$ must satisfy \eqref{eq:x*}, i.e., 
\begin{align}\label{eq:x*(p)}
    \nabla f_i(\boldsymbol{x}^*_i(\boldsymbol{p}_i))+\boldsymbol{p}_i=0\,.
\end{align}
Since $\nabla^{-1} f_i$ is well-defined, we are allowed to represent $\boldsymbol{x}^*_i(\boldsymbol{p}_i)$ in \eqref{eq:x*(p)} as \eqref{eq:demand-formula}, from which it is easy to see that $\boldsymbol{x}^*_i(\boldsymbol{p}_i)$ is a bijection since $\nabla^{-1} f_i$ is a bijection. 

Moreover, an important property of any bijective function is that it is one-to-one, which means that every element in the codomain is mapped to by at most one element in the domain. Thus, $\forall \tilde{\boldsymbol{p}}_i, \hat{\boldsymbol{p}}_i \in \real$, if
$\boldsymbol{x}^*_i(\tilde{\boldsymbol{p}}_i)=\boldsymbol{x}^*_i(\hat{\boldsymbol{p}}_i)$, then $\tilde{\boldsymbol{p}}_i=\hat{\boldsymbol{p}}_i$, which is logically equivalent to the contrapositive, i.e.,
$\forall \tilde{\boldsymbol{p}}_i, \hat{\boldsymbol{p}}_i \in \real$, if $\tilde{\boldsymbol{p}}_i\neq\hat{\boldsymbol{p}}_i$, then $\boldsymbol{x}^*_i(\tilde{\boldsymbol{p}}_i)\neq\boldsymbol{x}^*_i(\hat{\boldsymbol{p}}_i)$.  

Finally, we would like to show that $\boldsymbol{x}^*_i(\boldsymbol{p}_i)$ is a continuously differentiable function. By Assumption~\ref{ass:cost}, $f_i(\boldsymbol{x}_i)$ is twice continuously differentiable, which implies that $\nabla f_i(\boldsymbol{x}_i)$ is continuously differentiable. That is, $\nabla^2 f_i(\boldsymbol{x}_i)$ exists everywhere and is continuous. As mentioned in the beginning of the proof, $\nabla f_i(\boldsymbol{x}_i)$ is strictly increasing, which implies that $\nabla^2 f_i(\boldsymbol{x}_i)>0$ everywhere. By inverse function theorem, $\nabla^{-1} f_i$ is continuously differentiable and its derivative at $(-\boldsymbol{p}_i)$ is given by $1/\nabla^2 f_i(\nabla^{-1} f_i(-\boldsymbol{p}_i))=1/\nabla^2 f_i(\boldsymbol{x}^*_i(\boldsymbol{p}_i))>0$ since $\nabla^2 f_i(\boldsymbol{x}_i)>0$ everywhere. Thus, $\nabla\boldsymbol{x}^*_i(\boldsymbol{p}_i)=-1/\nabla^2 f_i(\nabla^{-1} f_i(-\boldsymbol{p}_i))=-1/\nabla^2 f_i(\boldsymbol{x}^*_i(\boldsymbol{p}_i))<0$, which is clearly continuous since $\nabla^2 f_i(\boldsymbol{x}_i)$ is continuous. This concludes the proof that $\boldsymbol{x}^*_i(\boldsymbol{p}_i)$ is continuously differentiable and strictly decreasing.
\end{proof}

Lemma~\ref{lem:x-p-bijective} shows that the demand update $\boldsymbol{x}^*_i(\boldsymbol{p}_i)$ is a bijective function which naturally enjoys a nice property. That is, $\forall \tilde{\boldsymbol{p}}_i, \hat{\boldsymbol{p}}_i \in \real$, if $\tilde{\boldsymbol{p}}_i\neq\hat{\boldsymbol{p}}_i$, then $\boldsymbol{x}^*_i(\tilde{\boldsymbol{p}}_i)\neq\boldsymbol{x}^*_i(\hat{\boldsymbol{p}}_i)$, which means that it is impossible for distinct price signals to motivate the same power demand. As the analysis will unfold later, this ``uniqueness" plays a role in the convergence of the pricing mechanism which we will propose.

Therefore, our goal is to design a suitable update for price profile $\vec{\boldsymbol{p}}:=\left(\boldsymbol{p}_i, i \in \mathcal{N} \right) \in \real^{n}$ which can leverage the nice demand update $\vec{\boldsymbol{x}}^*(\vec{\boldsymbol{p}}):=\left(\boldsymbol{x}^*_i(\boldsymbol{p}_i), i \in \mathcal{N} \right) \in \real^{n}$ induced by individual user's optimization problem \eqref{eq:opt-user} in each iteration to gradually gear the demand profile $\vec{\boldsymbol{x}}$ of users toward the minimizer of \eqref{eq:opt-plan} after enough iterations. This incentive pricing mechanism allows the system operator to achieve the desired solution to the planner's optimization problem \eqref{eq:opt-plan} without solving it directly.
\section{Adaptive Price Update under Nonsmoothness }
In terms of incentive pricing mechanism, our recent work \cite{LI2024EPSR} proposes for a similar but simpler setting to adopt price dynamics utilizing the gradient information of the system cost function $J(\cdot)$ to incentivize users to adjust their power consumption towards a point where the planner's problem and user's problem are simultaneously solved. However, the underlying assumption that $J(\cdot)$ is smooth  does not hold in our case due to the particular choice of $J(\cdot)$ as \eqref{eq:Jcostfun} which makes $J(\cdot)$ convex and locally Lipschitz but not differentiable everywhere. Thus, we propose an adaptive price update by leveraging the generalized gradient in Definition~\ref{def:def-generalized gra} as follows:
\begin{align}\label{eq:pupdate}   \dot{\vec{\boldsymbol{p}}}\in \partial J(\vec{\boldsymbol{x}}^*(\vec{\boldsymbol{p}})) -\vec{\boldsymbol{p}}\,.
\end{align}
This is well-defined since the fact that $J(\vec{\boldsymbol{x}})$ is a locally Lipschitz continous function ensures that $J(\vec{\boldsymbol{x}})$ has a nonempty compact set as its generalized gradient at any $\vec{\boldsymbol{x}}\in \real^{n}$~\cite[Proposition 6]{Cortes2008CSM}.

Based on \eqref{eq:pupdate}, we now illustrate the incentive pricing mechanism in more detail. As shown in Fig.~\ref{fig:Incentive_mechanism}, we consider a two-timescale design of incentive pricing mechanism, where individual users solve \eqref{eq:opt-user} for $\boldsymbol{x}^*_i(\boldsymbol{p}_i)$ much faster than the system operator updates the price $\vec{\boldsymbol{p}}$ via \eqref{eq:pupdate}. This timescale separation allows users to consider the price signal $\vec{\boldsymbol{p}}$ as static when solving for $\vec{\boldsymbol{x}}^*(\vec{\boldsymbol{p}})$. Thus, following any given price $\vec{\boldsymbol{p}}$ provided by the system operator, users adjust their power consumption towards $\vec{\boldsymbol{x}}^*(\vec{\boldsymbol{p}})$ almost immediately by solving \eqref{eq:x*} and then the system operator updates the price $\vec{\boldsymbol{p}}$ according to \eqref{eq:pupdate} in response to the current demand profile $\vec{\boldsymbol{x}}^*(\vec{\boldsymbol{p}})$. It should be intuitively clear that \eqref{eq:pupdate} provides users with incentives to align their own interests with social welfare, given that adjustments to $\vec{\boldsymbol{p}}$ intend to reduce the difference between the marginal cost of the individual user quantified by $\vec{\boldsymbol{p}}$ and the marginal cost of the system operator characterized by $\partial J(\vec{\boldsymbol{x}}^*(\vec{\boldsymbol{p}}))$.

\begin{figure}[h]
\centering
\includegraphics[width=\columnwidth]{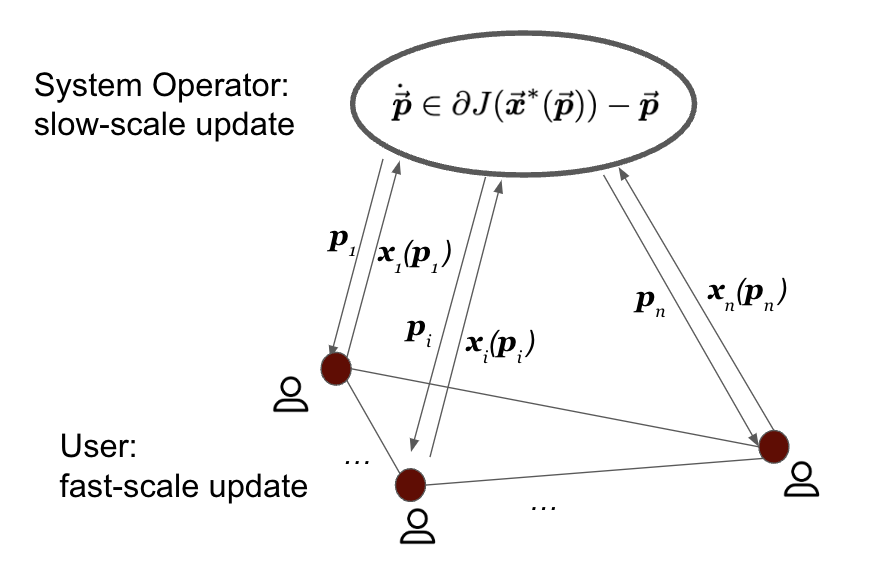}
\caption{Two-timescale design of incentive pricing mechanism.}\label{fig:Incentive_mechanism}
\end{figure}

With this in mind, as the system operator iteratively updates the price $\vec{\boldsymbol{p}}$, the nonsmooth dynamical system composed of \eqref{eq:x*} and \eqref{eq:pupdate} ideally should settle down at a point that achieves the optimal solution to the planner's optimization problem \eqref{eq:opt-plan}. That is, at the equilibrium price $\vec{\boldsymbol{p}}^\star:=\left(\boldsymbol{p}_i^\star, i \in \mathcal{N} \right) \in \real^{n}$, we would like to have $\vec{\boldsymbol{x}}^*(\vec{\boldsymbol{p}}^\star)$ satisfy \eqref{eq:optimal-planner}, which is captured by the following theorem.
\begin{thm}[Unique equilibrium with incentive aligned]\label{thm:eq-point} Under Assumption~\ref{ass:cost}, the demand profile $\vec{\boldsymbol{x}}^*(\vec{\boldsymbol{p}}^\star)$ occurring at the unique equilibrium price $\vec{\boldsymbol{p}}^\star$ of the dynamical system composed of \eqref{eq:x*} and \eqref{eq:pupdate} is the unique global minimizer to the planner's optimization problem \eqref{eq:opt-plan}, i.e., 
\begin{align}\label{eq:eq-mechanism}
   \boldsymbol{0}\in \partial J(\vec{\boldsymbol{x}}^*(\vec{\boldsymbol{p}}^\star))+\left(\nabla f_i(\boldsymbol{x}^*_i(\boldsymbol{p}_i^\star)), i \in \mathcal{N} \right)\,.
\end{align}    
\end{thm}
\begin{proof}
The point $\vec{\boldsymbol{p}}^\star$ is an equilibrium of the price update \eqref{eq:pupdate} if and only if ~\cite[Chapter 4.4]{Clarke1998Nonsmooth}
\begin{align}\label{eq:suff-pstar}
    \boldsymbol{0}\in \partial J(\vec{\boldsymbol{x}}^*(\vec{\boldsymbol{p}}^\star)) -\vec{\boldsymbol{p}}^\star\,.
\end{align}
Note that $\vec{\boldsymbol{x}}^*(\vec{\boldsymbol{p}}^\star)$ generated from the demand update satisfies \eqref{eq:x*}, i.e.,
\begin{align*}
    \nabla f_i(\boldsymbol{x}^*_i(\boldsymbol{p}_i^\star))+\boldsymbol{p}^\star_i=0\,,\qquad\forall i\in\mathcal{N}\,,
\end{align*}
from which we know 
\begin{align}\label{eq:x*atpstar}
    \vec{\boldsymbol{p}}^\star=-\left(\nabla f_i(\boldsymbol{x}^*_i(\boldsymbol{p}_i^\star)), i \in \mathcal{N} \right)\,.
\end{align}
Substituting \eqref{eq:x*atpstar} into \eqref{eq:suff-pstar} yields \eqref{eq:eq-mechanism},
which is exactly in the form of the optimality condition \eqref{eq:optimal-planner} for the planner's optimization problem \eqref{eq:opt-plan}. Thus, $\vec{\boldsymbol{x}}^*(\vec{\boldsymbol{p}}^\star)$ corresponding to the equilibrium price $\vec{\boldsymbol{p}}^\star$ is the unique global minimizer to \eqref{eq:opt-plan}. Now, it remains to show that the equilibrium price $\vec{\boldsymbol{p}}^\star$ is unique. By way of contradiction, suppose that both $\vec{\boldsymbol{p}}^\star$ and $\vec{\boldsymbol{p}}^\circ$ satisfy \eqref{eq:suff-pstar}, where $\vec{\boldsymbol{p}}^\star\neq\vec{\boldsymbol{p}}^\circ$. Then, by a similar argument as above, we know that both $\vec{\boldsymbol{x}}^*(\vec{\boldsymbol{p}}^\star)$ and $\vec{\boldsymbol{x}}^*(\vec{\boldsymbol{p}}^\circ)$ must satisfy the optimality condition \eqref{eq:optimal-planner}. Thus, $\vec{\boldsymbol{x}}^*(\vec{\boldsymbol{p}}^\star)$ and $\vec{\boldsymbol{x}}^*(\vec{\boldsymbol{p}}^\circ)$ are both optimizers of problem \eqref{eq:opt-plan}. Notably, by Lemma~\ref{lem:x-p-bijective}, our assumption $\vec{\boldsymbol{p}}^\star\neq\vec{\boldsymbol{p}}^\circ$ directly implies $\vec{\boldsymbol{x}}^*(\vec{\boldsymbol{p}}^\star)\neq\vec{\boldsymbol{x}}^*(\vec{\boldsymbol{p}}^\circ)$. Therefore, we now reach a situation where there are two distinct optimizers to problem \eqref{eq:opt-plan}, which contradicts the fact that problem \eqref{eq:opt-plan} has a unique minimizer. This concludes the proof of the uniqueness of the equilibrium price $\vec{\boldsymbol{p}}^\star$.
\end{proof}
Theorem~\ref{thm:eq-point} verifies that the proposed incentive pricing mechanism is guaranteed to settle down at a unique equilibrium price $\vec{\boldsymbol{p}}^\star$ whose corresponding demand profile $\vec{\boldsymbol{x}}^*(\vec{\boldsymbol{p}}^\star)$ is exactly the unique global minimizer to the planner's optimization problem \eqref{eq:opt-plan}. In other words, by adopting the proposed adaptive price update, the system operator can encourage users to align their individual benefits with the social welfare. Thus, the system objective of economic dispatch is achieved without disclosure of user privacy.

\section{Nonsmooth Stability Analysis}
Having characterized the equilibrium point and confirmed the incentive alignment at that point, we are now ready to investigate the stability of the nonsmooth dynamical system composed of the demand update \eqref{eq:x*} and the price update \eqref{eq:pupdate} by performing the natural extension of Lyapunov stability analysis provided in~\cite[Theorem 1]{Cortes2008CSM}. More precisely, the stability under the incentive pricing mechanism can be certified by finding a well-defined Lyapunov function that is decreasing along the trajectories of the system comprising~\eqref{eq:x*} and \eqref{eq:pupdate}. 
The main result of this section is presented below, whose proof is enabled by
a sequence of intermediate results that we discuss next.

\begin{thm}[Asymptotic stability]\label{thm:as-stable} Under Assumption~\ref{ass:cost}, the dynamical system composed of \eqref{eq:x*} and \eqref{eq:pupdate} is strongly asymptotically stable at the unique equilibrium characterized by \eqref{eq:eq-mechanism}. 
\end{thm}

Of course, before showing the stability of the system, we need to show that the dynamical system has a solution. Here, we take the solution to be in the Caratheodory sense, which roughly says that there is a trajectory that satisfies \eqref{eq:pupdate} except for a set of $t$ that has Lebesgue measure zero~\cite{Cortes2008CSM}. We do this by checking the conditions in \cite[Proposition S2]{Cortes2008CSM}. We use $\mathcal{B}(\mathbb{R}^d)$ to denote the collection of all subsets of $\mathbb{R}^d$ and $B(\mathbf{a},r)$ to denote the ball centered at $\mathbf{a}$ with radius $r$.  

Since $\partial J(\vec{\boldsymbol{x}}^*(\vec{\boldsymbol{p}}))$ involves the composition of functions, we develop the following lemma to facilitate our analysis.
\begin{lem}[Property preservation in composition]\label{lem:comp}
Assume that $h:\real^m\mapsto \real^d$ is continuous at $\boldsymbol{z}\in\real^m$ and $u$ is the composite of $h$ and $g:\real^d\mapsto \mathcal{B}(\real^d)$ defined by 
\begin{align}\label{eq:u-def}
    u(\boldsymbol{z}):=(g\circ h)(\boldsymbol{z}):=g(h(\boldsymbol{z}))\,.
\end{align}
\begin{itemize}
\item If $g:\real^d\mapsto \mathcal{B}(\real^d)$ is upper semicontinuous at $h(\boldsymbol{z})\in\real^d$, then $u$ is upper semicontinuous at $\boldsymbol{z}$. 
\item If $g:\real^d\mapsto \mathcal{B}(\real^d)$ is locally bounded at $h(\boldsymbol{z})\in\real^d$, then $u$ is locally bounded at $\boldsymbol{z}$.
\end{itemize}
\end{lem}
\begin{proof}
We study the two cases separately.

For upper semicontinuity of a set-valued map~\cite{Cortes2008CSM}, we need to show that, $\forall \epsilon>0$, $\exists\delta>0$ such that 
\begin{align}\label{eq:upper-u}
    u(\tilde{\boldsymbol{z}})\subset u(\boldsymbol{z})+B(\boldsymbol{0};\epsilon)\,,\qquad\forall\tilde{\boldsymbol{z}}\in B(\boldsymbol{z};\delta)\,.
\end{align}

To this end, we first note that, if $g$ is upper semicontinuous at $h(\boldsymbol{z})\in\real^d$, for any given $\epsilon>0$, $\exists\eta>0$ such that, whenever $\boldsymbol{y}\in B(h(\boldsymbol{z});\eta)$, it holds that~\cite{Cortes2008CSM}
\begin{align}\label{eq:gupper}
    g(\boldsymbol{y})\subset g(h(\boldsymbol{z}))+B(\boldsymbol{0};\epsilon)\,.
\end{align}
Next, since $h$ is continuous at $\boldsymbol{z}\in\real^m$, for any given $\eta>0$, $\exists\delta>0$ such that, whenever $\tilde{\boldsymbol{z}}\in B(\boldsymbol{z};\delta)$,  it holds that~\cite[Definition 4.5]{Rudin2013PMA} 
\begin{align*}
    h(\tilde{\boldsymbol{z}})\in B(h(\boldsymbol{z});\eta)\,.
\end{align*}
Now, we combine the above two arguments by setting $\boldsymbol{y}=h(\tilde{\boldsymbol{z}})$ in \eqref{eq:gupper}, which yields  
\begin{align}\label{eq:ghupper}
    g(h(\tilde{\boldsymbol{z}}))\subset g(h(\boldsymbol{z}))+B(\boldsymbol{0};\epsilon)\,, \qquad\forall\tilde{\boldsymbol{z}}\in B(\boldsymbol{z};\delta)\,.
\end{align}
Finally, from \eqref{eq:u-def}, we know $g(h(\tilde{\boldsymbol{z}}))=u(\tilde{\boldsymbol{z}})$ and $g(h(\boldsymbol{z}))=u(\boldsymbol{z})$, which combined with \eqref{eq:ghupper} gives exactly the claim \eqref{eq:upper-u} that we would like to prove.

For local boundedness of a set-valued map~\cite{Cortes2008CSM}, we need to show that, $\exists\delta>0$ and some constant $M>0$ such that 
\begin{align}\label{eq:com-bound}
    \|\boldsymbol{\mu}\|_2\leq M\,,\qquad\forall \tilde{\boldsymbol{z}}\in B(\boldsymbol{z};\delta),\boldsymbol{\mu}\in u(\tilde{\boldsymbol{z}})\,.
\end{align}

With this aim, we first note that, if $g$ is locally bounded at $h(\boldsymbol{z})\in\real^d$, $\exists\eta>0$ and some constant $M>0$ such that~\cite{Cortes2008CSM}
\begin{align}\label{eq:bound-g}
    \|\boldsymbol{\mu}\|_2\leq M\,,\qquad \forall\boldsymbol{y}\in B(h(\boldsymbol{z});\eta), \boldsymbol{\mu}\in g(\boldsymbol{y})\,.
\end{align}
Again, since $h$ is continuous at $\boldsymbol{z}\in\real^m$, for any given $\eta>0$, $\exists\delta>0$ such that~\cite[Definition 4.5]{Rudin2013PMA}
\begin{align}\label{eq:h-con}
    h(\tilde{\boldsymbol{z}})\in B(h(\boldsymbol{z});\eta)\,,\qquad\forall\tilde{\boldsymbol{z}}\in B(\boldsymbol{z};\delta)\,.
\end{align}
Now, we combine the above two arguments by setting $\boldsymbol{y}=h(\tilde{\boldsymbol{z}})$ in \eqref{eq:bound-g}, which yields 
\begin{align*}
    \|\boldsymbol{\mu}\|_2\leq M\,,\forall\tilde{\boldsymbol{z}}\in\! B(\boldsymbol{z};\delta), h(\tilde{\boldsymbol{z}})\in\! B(h(\boldsymbol{z});\eta), \boldsymbol{\mu}\in g(h(\tilde{\boldsymbol{z}}))\,.
\end{align*}
Here, the second condition can be removed since it is directly implied by the first condition due to \eqref{eq:h-con}, which yields
\begin{align}\label{eq:bound-u}
    \|\boldsymbol{\mu}\|_2\leq M\,,\quad \forall\tilde{\boldsymbol{z}}\in B(\boldsymbol{z};\delta), \boldsymbol{\mu}\in g(h(\tilde{\boldsymbol{z}}))\,.
\end{align}
Finally, from \eqref{eq:u-def}, we know $g(h(\tilde{\boldsymbol{z}}))=u(\tilde{\boldsymbol{z}})$, which substituted into \eqref{eq:bound-u} gives exactly the claim \eqref{eq:com-bound} that we would like to prove.
\end{proof}
Lemma~\ref{lem:comp} paves us a way to show the existence of a Caratheodory solution of our dynamical system by checking conditions in~\cite[Propostion S2]{Cortes2008CSM}, which is the core of the next lemma.
\begin{lem}[Existence of a Caratheodory solution]\label{lem:exist} Under Assumption~\ref{ass:cost}, there exists a Caratheodory solution of the dynamical system composed of \eqref{eq:x*} and \eqref{eq:pupdate} for any initial condition $\vec{\boldsymbol{p}}(0)$.
\end{lem}
\begin{proof}
Basically, by~\cite[Propostion S2]{Cortes2008CSM}, it suffices to show that the set-valued map $\vec{\boldsymbol{p}}\mapsto[\partial J(\vec{\boldsymbol{x}}^*(\vec{\boldsymbol{p}})) -\vec{\boldsymbol{p}}]$ takes nonempty compact convex values and is also upper semicontinuous as well as locally bounded\footnote{There is not need to check measurability here since \eqref{eq:pupdate} does not explicitly depend on time $t$.}.

Clearly, it is the term $\partial J(\vec{\boldsymbol{x}}^*(\vec{\boldsymbol{p}}))$ associated with the generalized gradient in the above mapping that makes our dynamics different from differential equations. Thus, we focus our analysis on properties of $\partial J(\vec{\boldsymbol{x}}^*(\vec{\boldsymbol{p}}))$, which is a composition of $\partial J(\vec{\boldsymbol{x}})$ and $\vec{\boldsymbol{x}}^*(\vec{\boldsymbol{p}})$.

   We start by investigating $\partial J(\vec{\boldsymbol{x}})$. Based on~\cite[Proposition 6]{Cortes2008CSM}, it follows directly from the local Lipschitz continuity of $J(\vec{\boldsymbol{x}})$ that $\partial J(\vec{\boldsymbol{x}})$ is a nonempty compact convex set at any $\vec{\boldsymbol{x}}$ and the set-valued map $\vec{\boldsymbol{x}}\mapsto\partial J(\vec{\boldsymbol{x}})$ is upper semicontinuous and locally bounded  at any $\vec{\boldsymbol{x}}$. 

   As for $\vec{\boldsymbol{x}}^*(\vec{\boldsymbol{p}}):=\left(\boldsymbol{x}^*_i(\boldsymbol{p}_i), i \in \mathcal{N} \right)$, it is a continuous vector-valued function since each component $\boldsymbol{x}^*_i(\boldsymbol{p}_i)$ is a continuous function by Lemma~\ref{lem:x-p-bijective}~\cite[Theorem 2.4]{Lax2017MultiCal}.

   With above information about $\partial J(\vec{\boldsymbol{x}})$ and $\vec{\boldsymbol{x}}^*(\vec{\boldsymbol{p}})$, we are now ready to examine properties of $\partial J(\vec{\boldsymbol{x}}^*(\vec{\boldsymbol{p}}))$. First, given that $\partial J(\vec{\boldsymbol{x}})$ is a nonempty compact convex set at any $\vec{\boldsymbol{x}}$, it must be true that $\partial J(\vec{\boldsymbol{x}}^*(\vec{\boldsymbol{p}}))$ is a nonempty compact convex set at any $\vec{\boldsymbol{p}}$ as well since $\vec{\boldsymbol{x}}^*(\vec{\boldsymbol{p}})$ is a bijective function by Lemma~\ref{lem:x-p-bijective}. This can be understood by noting that, no matter what particular value the price signal $\vec{\boldsymbol{p}}$ is taking, $\partial J(\vec{\boldsymbol{x}})$ will take a corresponding value $\vec{\boldsymbol{x}}=\vec{\boldsymbol{x}}^*(\vec{\boldsymbol{p}})$ at that $\vec{\boldsymbol{p}}$, which must produce a nonempty compact convex set $\partial J(\vec{\boldsymbol{x}}^*(\vec{\boldsymbol{p}}))$. Second, the upper semicontinuity and local boundedness of $\partial J(\vec{\boldsymbol{x}}^*(\vec{\boldsymbol{p}}))$ at any $\vec{\boldsymbol{p}}$ follow from Lemma~\ref{lem:comp} by setting $h=\vec{\boldsymbol{x}}^*(\vec{\boldsymbol{p}})$ which is continuous and $g=\partial J(\vec{\boldsymbol{x}})$ which is upper semi-continuous and locally bounded.

   Finally, the term $(-\vec{\boldsymbol{p}})$ has no influence to the above properties. First, it only translates the nonempty compact convex set $\partial J(\vec{\boldsymbol{x}}^*(\vec{\boldsymbol{p}}))$ by $(-\vec{\boldsymbol{p}})$, which is still a nonempty compact convex set. Thus, $\vec{\boldsymbol{p}}\mapsto[\partial J(\vec{\boldsymbol{x}}^*(\vec{\boldsymbol{p}})) -\vec{\boldsymbol{p}}]$ takes nonempty compact convex values. Second, it can be considered as a continuous function which is inherently upper semicontinuous and locally bounded at $\vec{\boldsymbol{p}}$. Since the summation of two upper semicontinuous functions is still upper semicontinuous and the summation of two locally bounded functions is still locally bounded. Thus, $\vec{\boldsymbol{p}}\mapsto[\partial J(\vec{\boldsymbol{x}}^*(\vec{\boldsymbol{p}})) -\vec{\boldsymbol{p}}]$ is also upper semicontinuous and locally bounded. The result follows from~\cite[Propostion S2]{Cortes2008CSM}.
\end{proof}
After determining the existence of a Caratheodory solution of the system from any initial point through Lemma~\ref{lem:exist}, we now examine the nonsmooth system stability by constructing a candidate Lyapunov function. We seek a function $V(\vec{\boldsymbol{p}})$ that is locally Lipschitz and regular and also satisfies $V(\vec{\boldsymbol{p}}^\star)=0$ and $V(\vec{\boldsymbol{p}})>0$, $\forall\vec{\boldsymbol{p}}\neq\vec{\boldsymbol{p}}^\star$. The monotonicity of this Lyapunov candidate along the system trajectories is more complicated compared to a standard analysis since we need to study the Lie derivative in a nonsmooth setting. 

We consider the following Lyapunov function candidate:
\begin{align}  \label{eq:V-def}
V(\vec{\boldsymbol{p}}):=C(\vec{\boldsymbol{x}}^*(\vec{\boldsymbol{p}}))-C(\vec{\boldsymbol{x}}^*(\vec{\boldsymbol{p}}^\star))\,,
\end{align}
where $C(\cdot)$ denotes the objective function of the planner's optimization problem \eqref{eq:opt-plan} and $\vec{\boldsymbol{p}}^\star$ corresponds to the unique equilibrium point of the system satisfying \eqref{eq:eq-mechanism}. The next result shows that this is a well-defined Lyapunov function
candidate. 
\begin{lem}[Well-defined Lyapunov function]\label{lem:V}
Under Assumption~\ref{ass:cost},  $V(\vec{\boldsymbol{p}})$ defined in \eqref{eq:V-def} is a locally Lipschitz and regular function that satisfies
$V(\vec{\boldsymbol{p}}^\star)=0$ and $V(\vec{\boldsymbol{p}})>0$, $\forall\vec{\boldsymbol{p}}\neq\vec{\boldsymbol{p}}^\star$.
\end{lem}
\begin{proof} As discussed in Section~\ref{ssec:plan}, the entire objective function $C(\vec{\boldsymbol{x}})$ of the planner's optimization problem \eqref{eq:opt-plan} is locally Lipschitz and strictly convex, which together with the continuous differentiability of each $\boldsymbol{x}^*_i(\boldsymbol{p}_i)$ by Lemma~\ref{lem:x-p-bijective} allows us to show that $V(\vec{\boldsymbol{p}})$ is locally Lipschitz and regular. We now illustrate this in detail. 

We begin with the local Lipschitz continuity $C(\vec{\boldsymbol{x}}^*(\vec{\boldsymbol{p}}))$. 
Clearly, the continuous differentiability of each $\boldsymbol{x}^*_i(\boldsymbol{p}_i)$ implies that each $\boldsymbol{x}^*_i(\boldsymbol{p}_i)$ is locally Lipschitz~\cite[Chapter 17.2]{Hirsch2013Hirsch}.
Now, since each component of $\vec{\boldsymbol{x}}^*(\vec{\boldsymbol{p}}):=\left(\boldsymbol{x}^*_i(\boldsymbol{p}_i), i \in \mathcal{N} \right)$ is locally Lipschitz and $C(\vec{\boldsymbol{x}})$ is locally Lipschitz as well, their composition $C(\vec{\boldsymbol{x}}^*(\vec{\boldsymbol{p}}))$ is locally Lipschitz by the chain rule~\cite{Cortes2008CSM}.

We next investigate the regularity of $C(\vec{\boldsymbol{x}}^*(\vec{\boldsymbol{p}}))$. First of all, $\vec{\boldsymbol{x}}^*(\vec{\boldsymbol{p}}):=\left(\boldsymbol{x}^*_i(\boldsymbol{p}_i), i \in \mathcal{N} \right)$ is a continuously differentiable vector-valued function since each component $\boldsymbol{x}^*_i(\boldsymbol{p}_i)$ is a continuously differentiable function~\cite[Theorem 2.8]{Spivak1965Calculus}. Moreover, $C(\vec{\boldsymbol{x}})$ is locally Lipschitz and strictly convex, which further ensures that $C(\vec{\boldsymbol{x}})$ is regular~\cite[Proposition 2.4.3]{Clarke1998Nonsmooth}. By~\cite[Theorem 8.18]{Clason2024Nonsmooth}, as a composite of $\vec{\boldsymbol{x}}^*(\vec{\boldsymbol{p}})$ and $C(\vec{\boldsymbol{x}})$, $C(\vec{\boldsymbol{x}}^*(\vec{\boldsymbol{p}}))$ is regular. Hence, $V(\vec{\boldsymbol{p}})$ is locally Lipschitz and regular since the other term in $V(\vec{\boldsymbol{p}})$ in \eqref{eq:V-def} is a constant.

Clearly, $V(\vec{\boldsymbol{p}}^\star)=0$ by construction. To see why $V(\vec{\boldsymbol{p}})>0$, $\forall\vec{\boldsymbol{p}}\neq\vec{\boldsymbol{p}}^\star$, we first note that $\vec{\boldsymbol{x}}^*(\vec{\boldsymbol{p}}^\star)$ is the unique global minimizer to problem \eqref{eq:opt-plan} by Theorem~\ref{thm:eq-point}, which means that
\begin{align}\label{eq:uni-xstar-ineq}
\forall\vec{\boldsymbol{x}}\neq\vec{\boldsymbol{x}}^*(\vec{\boldsymbol{p}}^\star)\,,\qquad C(\vec{\boldsymbol{x}})>C(\vec{\boldsymbol{x}}^*(\vec{\boldsymbol{p}}^\star))\,.   
\end{align}
Moreover, we know from Lemma~\ref{lem:x-p-bijective} that, $\forall\vec{\boldsymbol{p}}\neq\vec{\boldsymbol{p}}^\star$, it holds that $\vec{\boldsymbol{x}}^*(\vec{\boldsymbol{p}})\neq\vec{\boldsymbol{x}}^*(\vec{\boldsymbol{p}}^\star)$. Therefore, setting $\vec{\boldsymbol{x}}=\vec{\boldsymbol{x}}^*(\vec{\boldsymbol{p}})$ in \eqref{eq:uni-xstar-ineq}, we get $C(\vec{\boldsymbol{x}}^*(\vec{\boldsymbol{p}}))>C(\vec{\boldsymbol{x}}^*(\vec{\boldsymbol{p}}^\star))$, i.e., $C(\vec{\boldsymbol{x}}^*(\vec{\boldsymbol{p}}))-C(\vec{\boldsymbol{x}}^*(\vec{\boldsymbol{p}}^\star))>0$,
which is equivalent to 
$V(\vec{\boldsymbol{p}})>0$ by our construction of $V(\vec{\boldsymbol{p}})$ in \eqref{eq:V-def}. This confirms that $V(\vec{\boldsymbol{p}})>0$,
$\forall\vec{\boldsymbol{p}}\neq\vec{\boldsymbol{p}}^\star$, as desired.
\end{proof}
Next, we  turn to verify the monotonic evolution of $V(\vec{\boldsymbol{p}})$ along the system trajectories given by the notion of Lie derivative in the nonsmooth setting, which requires $\max \tilde{\mathcal{L}}V(\vec{\boldsymbol{p}})<0$, $\forall\vec{\boldsymbol{p}}\neq\vec{\boldsymbol{p}}^\star$, with $\tilde{\mathcal{L}}V(\vec{\boldsymbol{p}})$ being the set-valued Lie derivative of $V$ regarding $[\partial J(\vec{\boldsymbol{x}}^*(\vec{\boldsymbol{p}})) -\vec{\boldsymbol{p}}]$ in \eqref{eq:pupdate} at $\vec{\boldsymbol{p}}$ defined by~\cite{Cortes2008CSM, Shevitz1994tac}
\begin{align}\label{eq:def-lieV}
   \tilde{\mathcal{L}}V(\vec{\boldsymbol{p}}):=&\bigl\{a\in \real: \exists \boldsymbol{v}\in\partial J(\vec{\boldsymbol{x}}^*(\vec{\boldsymbol{p}})) -\vec{\boldsymbol{p}} \text{\ such that\ }\nonumber\\&\ \ \boldsymbol{\zeta}^T \boldsymbol{v}=a, \forall \boldsymbol{\zeta}\in \partial V(\vec{\boldsymbol{p}})\bigr\}\nonumber\\
=&\cap_{\boldsymbol{\zeta}\in \partial V(\vec{\boldsymbol{p}})} \boldsymbol{\zeta}^T [\partial J(\vec{\boldsymbol{x}}^*(\vec{\boldsymbol{p}})) -\vec{\boldsymbol{p}}]\,.   
\end{align}
\begin{lem}[Negativity of Lie derivative] \label{lem:lie}Under Assumption~\ref{ass:cost}, the set-valued Lie derivative of $V(\vec{\boldsymbol{p}})$ defined in \eqref{eq:def-lieV} satisfies $\max \tilde{\mathcal{L}}V(\vec{\boldsymbol{p}})<0$, $\forall\vec{\boldsymbol{p}}\neq\vec{\boldsymbol{p}}^\star$.
\end{lem}
\begin{proof}
Before delving into $\tilde{\mathcal{L}}V(\vec{\boldsymbol{p}})$, we need to characterize $\partial V(\vec{\boldsymbol{p}})$ which can be computed as
\begin{align}\label{eq:dV}
\!\!\!\partial V(\vec{\boldsymbol{p}})=&\ \partial (C\circ \vec{\boldsymbol{x}}^*)(\vec{\boldsymbol{p}})\nonumber\\
\stackrel{\circled{1}}{=}&\ \nabla \vec{\boldsymbol{x}}^*(\vec{\boldsymbol{p}})\partial C(\vec{\boldsymbol{x}}^*(\vec{\boldsymbol{p}}))\nonumber\\ :=&\left\{\nabla \vec{\boldsymbol{x}}^*(\vec{\boldsymbol{p}})\boldsymbol{\eta}:\boldsymbol{\eta}\in\partial C(\vec{\boldsymbol{x}}^*(\vec{\boldsymbol{p}}))\right\}\nonumber\\
=&\left\{\diag{\nabla \boldsymbol{x}^*_i(\boldsymbol{p}_i), i \!\in\! \mathcal{N} } \boldsymbol{\eta}:\boldsymbol{\eta}\!\in\!\partial C(\vec{\boldsymbol{x}}^*(\vec{\boldsymbol{p}}))\right\}\!\,.
\end{align}
In \circled{1}, the chain rule of the generalized gradient~\cite[Theorem 8.18]{Clason2024Nonsmooth} can be used with equality since, as discussed in the proof of Lemma~\ref{lem:V}, the conditions that $C(\vec{\boldsymbol{x}})$ is locally Lipschitz and regular and that $\vec{\boldsymbol{x}}^*(\vec{\boldsymbol{p}})$ is continuously differentiable both hold.

To get a more explicit expression for \eqref{eq:dV}, we now derive $\partial C(\vec{\boldsymbol{x}}^*(\vec{\boldsymbol{p}}))$ as  
\begin{align}\label{eq:dC}
    \partial C(\vec{\boldsymbol{x}}^*(\vec{\boldsymbol{p}}))=&\
\partial\left(\sum_{i \in \mathcal{N}} f_i(\boldsymbol{x}^*_i(\boldsymbol{p}_i)) +J(\vec{\boldsymbol{x}}^*(\vec{\boldsymbol{p}}))\right)\nonumber\\=&\left(\nabla f_i(\boldsymbol{x}^*_i(\boldsymbol{p}_i)), i \in \mathcal{N} \right)+\partial J(\vec{\boldsymbol{x}}^*(\vec{\boldsymbol{p}}))\nonumber\\=&\ \partial J(\vec{\boldsymbol{x}}^*(\vec{\boldsymbol{p}}))-\vec{\boldsymbol{p}}\,,
\end{align}
where the first equality is due to the definition of $C(\vec{\boldsymbol{x}})$ in \eqref{eq:opt-plan}, the second equality is is due to the sum rule of the generalized gradient for convex functions~\cite[Chapter 2.4]{Clarke1998Nonsmooth} as mentioned in Section~\ref{ssec:plan}, the last equality uses the relation that $\left(\nabla f_i(\boldsymbol{x}^*_i(\boldsymbol{p}_i)), i \in \mathcal{N} \right)=-\vec{\boldsymbol{p}}$ resulting from the optimality condition \eqref{eq:x*(p)} of user's problem as discussed in the proof of Lemma~\ref{lem:x-p-bijective}. 

Substituting \eqref{eq:dC} to \eqref{eq:dV} yields
\begin{align}\label{eq:dV-2}
\partial V(\vec{\boldsymbol{p}})=\bigl\{\diag{\nabla \boldsymbol{x}^*_i(\boldsymbol{p}_i), &i \in \mathcal{N} } \boldsymbol{\eta}:\nonumber\\&\boldsymbol{\eta}\in\partial J(\vec{\boldsymbol{x}}^*(\vec{\boldsymbol{p}}))-\vec{\boldsymbol{p}}\bigr\}\,,  
\end{align}
which will be applied to \eqref{eq:def-lieV} for investigating the sign of $\max \tilde{\mathcal{L}}V(\vec{\boldsymbol{p}})$. The challenge part is that $J(\cdot)$ is continuous but not differentiable everywhere, which means that there exists a set of points of $\vec{\boldsymbol{p}}$ for which $J(\cdot)$ fails to be differentiable at the corresponding $\vec{\boldsymbol{x}}^*(\vec{\boldsymbol{p}})$. For the ease of notation, we denote such a set of $\vec{\boldsymbol{p}}$ as $\Omega_{ J (\vec{\boldsymbol{x}}^*(\vec{\boldsymbol{p}}))}\subset\real^{n}$. Note that we only care about points $\vec{\boldsymbol{p}}$ different from the equilibrium $\vec{\boldsymbol{p}}^\star$ in this particular analysis, which are $\vec{\boldsymbol{p}}$ satisfying
\begin{align}\label{eq:suff-non-pstar}
    \boldsymbol{0}\notin \partial J(\vec{\boldsymbol{x}}^*(\vec{\boldsymbol{p}})) -\vec{\boldsymbol{p}}
\end{align}
by \eqref{eq:suff-pstar} in the proof of Theorem~\ref{thm:eq-point}.
This allows us to consider two cases based on whether $\vec{\boldsymbol{p}}\neq\vec{\boldsymbol{p}}^\star$ is in $\Omega_{ J (\vec{\boldsymbol{x}}^*(\vec{\boldsymbol{p}}))}$ or not. 
\begin{enumerate}
    \item If $\vec{\boldsymbol{p}}\neq\vec{\boldsymbol{p}}^\star$ and $\vec{\boldsymbol{p}}\notin\Omega_{ J (\vec{\boldsymbol{x}}^*(\vec{\boldsymbol{p}}))}$: The generalized gradient $\partial J(\vec{\boldsymbol{x}}^*(\vec{\boldsymbol{p}}))$ reduces to a singleton, i.e., 
    \begin{align}\label{eq:dJ-single}
        \partial J(\vec{\boldsymbol{x}}^*(\vec{\boldsymbol{p}}))=\left\{\nabla J(\vec{\boldsymbol{x}}^*(\vec{\boldsymbol{p}}))\right\}\,.
    \end{align}
    Thus, $\partial V(\vec{\boldsymbol{p}})$ in \eqref{eq:dV-2} reduces to a singleton as well, i.e.,   
\begin{align} \label{eq:dV-diff}
\partial V(\vec{\boldsymbol{p}})=\bigl\{\diag{\nabla \boldsymbol{x}^*_i(\boldsymbol{p}_i), i \in \mathcal{N} }F(\vec{\boldsymbol{p}})\bigr\} 
\end{align}

with $$F(\vec{\boldsymbol{p}})=\nabla J(\vec{\boldsymbol{x}}^*(\vec{\boldsymbol{p}}))-\vec{\boldsymbol{p}}$$
in this case, which together with \eqref{eq:dJ-single} simplifies \eqref{eq:def-lieV} to
\begin{align} \label{eq:LV-case1}
&\tilde{\mathcal{L}}V(\vec{\boldsymbol{p}})\nonumber\\=&\bigl\{a\in \real: \exists \boldsymbol{v}\!\in\!\left\{F(\vec{\boldsymbol{p}}) \right\} \text{such that\ }\boldsymbol{\zeta}^T \boldsymbol{v}=a, \nonumber\\&\ \ \forall \boldsymbol{\zeta}\in \bigl\{\diag{\nabla \boldsymbol{x}^*_i(\boldsymbol{p}_i), i \in \mathcal{N} }F(\vec{\boldsymbol{p}})\bigr\}\bigr\} \nonumber
 \\=&\bigl\{\left(\diag{\nabla \boldsymbol{x}^*_i(\boldsymbol{p}_i), i \in \mathcal{N} }F(\vec{\boldsymbol{p}})\right)^T F(\vec{\boldsymbol{p}})\bigr\}\nonumber\\
=&\bigl\{F(\vec{\boldsymbol{p}})^T\diag{(\nabla \boldsymbol{x}^*_i(\boldsymbol{p}_i)), i \in \mathcal{N} }F(\vec{\boldsymbol{p}})\bigr\},
\end{align}

We claim that 
\begin{align} \label{eq:LVvalue-case1}
F(\vec{\boldsymbol{p}})^T\diag{(\nabla \boldsymbol{x}^*_i(\boldsymbol{p}_i)), i \in \mathcal{N} }F(\vec{\boldsymbol{p}})<0  
\end{align}
for two reasons. First, $\diag{(\nabla \boldsymbol{x}^*_i(\boldsymbol{p}_i)), i \in \mathcal{N} }\prec0$ since each $\nabla \boldsymbol{x}^*_i(\boldsymbol{p}_i)<0$ by Lemma~\ref{lem:x-p-bijective}. 
Second, since $\vec{\boldsymbol{p}}\neq\vec{\boldsymbol{p}}^\star$, we know from \eqref{eq:suff-non-pstar} that $F(\vec{\boldsymbol{p}})\neq \boldsymbol{0}$. Then, \eqref{eq:LVvalue-case1} follows directly. Combining \eqref{eq:LV-case1} and \eqref{eq:LVvalue-case1}, we know 
\begin{align}\label{eq:max-LV-case1}
\max \tilde{\mathcal{L}}V(\vec{\boldsymbol{p}})=&\ F(\vec{\boldsymbol{p}})^T\diag{(\nabla \boldsymbol{x}^*_i(\boldsymbol{p}_i)), i \in \mathcal{N} }F(\vec{\boldsymbol{p}})\nonumber\\<&\ 0 \,.   
\end{align}
    \item If $\vec{\boldsymbol{p}}\neq\vec{\boldsymbol{p}}^\star$ and $\vec{\boldsymbol{p}}\in\Omega_{ J (\vec{\boldsymbol{x}}^*(\vec{\boldsymbol{p}}))}$: Substituting \eqref{eq:dV-2} to \eqref{eq:def-lieV} yields
\begin{align}\label{eq:LV-case2}
 &\tilde{\mathcal{L}}V(\vec{\boldsymbol{p}})\\=&\bigl\{a\in \real: \exists \boldsymbol{v}\!\in\!\partial J(\vec{\boldsymbol{x}}^*(\vec{\boldsymbol{p}})) -\vec{\boldsymbol{p}} \text{\ such that\ }\boldsymbol{\zeta}^T \boldsymbol{v}=a, \nonumber\\&\ \ \forall \boldsymbol{\zeta}\!\in \!\!\bigl\{\diag{\nabla \boldsymbol{x}^*_i(\boldsymbol{p}_i), i \!\in \!\mathcal{N} } \boldsymbol{\eta}\!:\boldsymbol{\eta}\!\in\!\partial J(\vec{\boldsymbol{x}}^*(\vec{\boldsymbol{p}}))\!-\!\vec{\boldsymbol{p}}\bigr\} \!\,.\nonumber
\end{align}
If $\tilde{\mathcal{L}}V(\vec{\boldsymbol{p}})=\emptyset$, then $\max \tilde{\mathcal{L}}V(\vec{\boldsymbol{p}})=-\infty$ by convention~\cite{Cortes2008CSM}. If $\tilde{\mathcal{L}}V(\vec{\boldsymbol{p}})\neq\emptyset$,
we claim that, $\forall a\in\tilde{\mathcal{L}}V(\vec{\boldsymbol{p}})$ in \eqref{eq:LV-case2}, $a<0$.  To see this, we first note that, for any such $a$, $\exists \boldsymbol{v}\!\in\!\partial J(\vec{\boldsymbol{x}}^*(\vec{\boldsymbol{p}})) -\vec{\boldsymbol{p}}$, such that  $\boldsymbol{\zeta}^T \boldsymbol{v}=a$, $\forall \boldsymbol{\zeta}\!\in \!\!\bigl\{\diag{\nabla \boldsymbol{x}^*_i(\boldsymbol{p}_i), i \!\in \!\mathcal{N} } \boldsymbol{\eta}\!:\boldsymbol{\eta}\!\in\!\partial J(\vec{\boldsymbol{x}}^*(\vec{\boldsymbol{p}}))\!-\!\vec{\boldsymbol{p}}\bigr\}$. Clearly, we can pick $\boldsymbol{\eta}=\boldsymbol{v}$ and then $\boldsymbol{\zeta}=\diag{\nabla \boldsymbol{x}^*_i(\boldsymbol{p}_i), i \!\in \!\mathcal{N} }\boldsymbol{v}$ to solve for 
\begin{align*}
a=\boldsymbol{\zeta}^T \boldsymbol{v}=(\diag{\nabla \boldsymbol{x}^*_i(\boldsymbol{p}_i), i \!\in \!\mathcal{N} }\boldsymbol{v})^T\boldsymbol{v}\nonumber\\ =\boldsymbol{v}^T\diag{(\nabla \boldsymbol{x}^*_i(\boldsymbol{p}_i)), i \!\in \!\mathcal{N} }\boldsymbol{v} <0\,.  
\end{align*}
Here, the inequality is due to a similar argument for the previous case. That is,  $\diag{(\nabla \boldsymbol{x}^*_i(\boldsymbol{p}_i)), i \in \mathcal{N} }\prec0$ and $\boldsymbol{v}\neq \boldsymbol{0}$ from \eqref{eq:suff-non-pstar}. Now, we have shown that, if $\tilde{\mathcal{L}}V(\vec{\boldsymbol{p}})\neq\emptyset$, $\forall a\in\tilde{\mathcal{L}}V(\vec{\boldsymbol{p}})$ in \eqref{eq:LV-case2}, $a<0$, which ensures that $\max \tilde{\mathcal{L}}V(\vec{\boldsymbol{p}})<0$. Therefore, no matter whether $\tilde{\mathcal{L}}V(\vec{\boldsymbol{p}})$ is empty or not, it must be true that $\max \tilde{\mathcal{L}}V(\vec{\boldsymbol{p}})<0$.
\end{enumerate}
In sum, $\max \tilde{\mathcal{L}}V(\vec{\boldsymbol{p}})<0$, $\forall\vec{\boldsymbol{p}}\neq\vec{\boldsymbol{p}}^\star$.
\end{proof}
According to~\cite[Theorem 1]{Cortes2008CSM}, with the aid of Lemmas~\ref{lem:exist}, \ref{lem:V}, and \ref{lem:lie}, we now have all the elements necessary to establish the strongly asymptotical stability of the unique equilibrium as summarized in Theorem~\ref{thm:as-stable}.

\section{Numerical Illustrations}
In this section, we present simulation results to show the convergence of our proposed incentive pricing mechanism to the desired optimal solution to the global social welfare problem. The simulations are conducted on the IEEE 14-bus system, which contains $5$ generators and $20$ transmission lines. The generation cost of each generator is assumed to be linear as in \eqref{eq:Jcostfun}, with cost coefficient $c_i$ generated uniformly at random from $[5,20]$. We assume that there is one user on each bus and each user $i$ has disutility $f_i(\boldsymbol{x}_i)=(\boldsymbol{x}_i-\bar{\boldsymbol{x}}_i)^2$, where $\bar{\boldsymbol{x}}_i$ is some constant, representing for example, the targeted consumption of user $i$. 

In order to achieve the optimal solution to the global social welfare problem \eqref{eq:opt-plan} without solving it directly, we randomly initialize the price signal for individual users from $[5,15]$ and run our proposed incentive pricing mechanism, whose dynamics together with the evolution of user demand profile are provided in Fig.~\ref{fig:p_x-t}. Obviously, both price signals $\vec{\boldsymbol{p}}$ and user demands $\vec{\boldsymbol{p}}$ converge very fast. Particularly, $\vec{\boldsymbol{p}}$ successfully converges to the optimal solution of problem \eqref{eq:opt-plan}. Note that the cost are each bus are converges to one of a few values, which is typical when a small number of lines are congested~\cite{zhang2014network}. 
\begin{figure}[ht]
\centering
\includegraphics[width=\columnwidth]{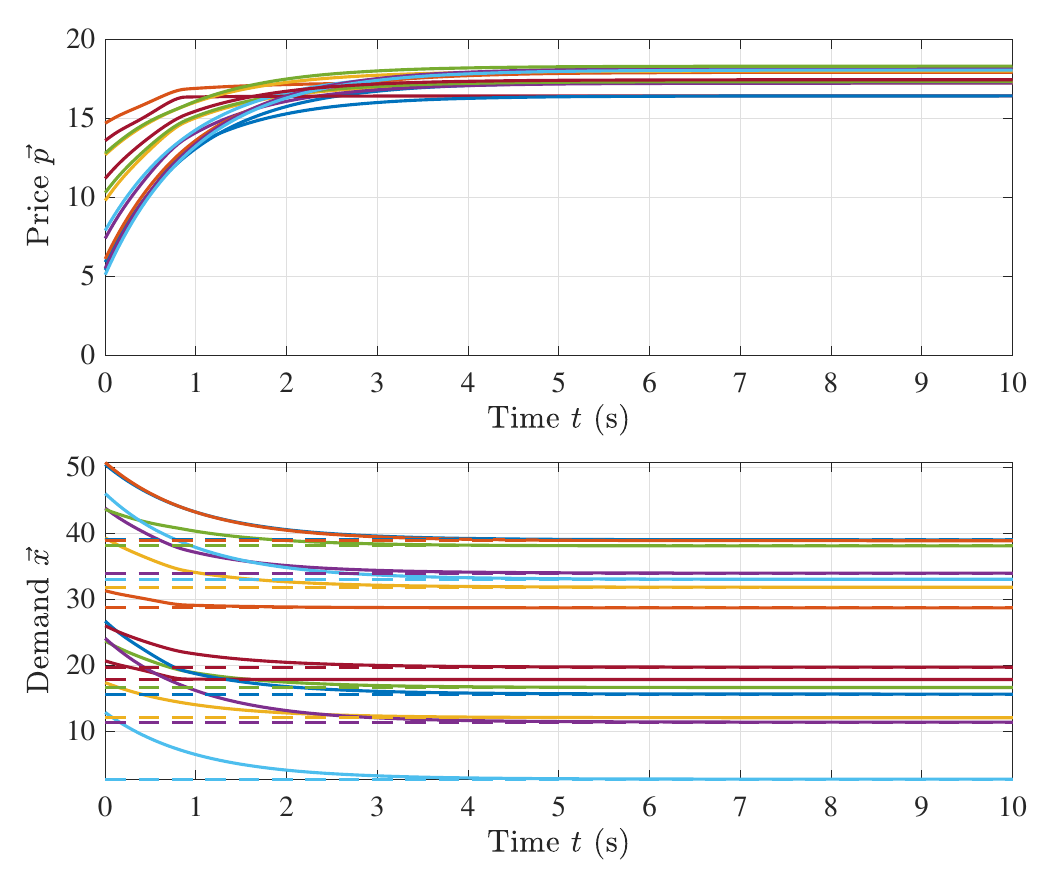}
\caption{Convergence of price signals and user demands of our proposed mechanism in IEEE 14-bus system with $14$ users, where the dashed lines represent the optimal demand profile to the global social welfare problem.}
\label{fig:p_x-t}
\end{figure}


\section{Conclusions and Outlook}
This paper extends adaptive pricing mechanisms for social welfare optimization to network-constrained energy systems with nonsmooth cost structures. By embedding DCOPF constraints into the operator’s objective and introducing a generalized gradient-based price update rule, we establish a provably convergent and privacy-preserving incentive design framework. Our theoretical analysis demonstrates the existence, uniqueness, and strong asymptotic stability of the equilibrium. Simulation results validate the practical effectiveness of the proposed mechanism in guiding user behavior toward globally optimal outcomes under realistic power network constraints.

Looking ahead, several important extensions remain open. First, practical systems are subject to uncertainty from renewable generation and stochastic user demand. Extending the current framework to handle uncertainty explicitly, either through robust or stochastic formulations of DCOPF, is a natural next step. Second, applying the method to AC power flow models would enhance its applicability to real-world grids, though this introduces significant nonconvexity. Third, while our current pricing update relies on analytical computation of generalized gradients, a promising direction is to develop data-driven or learning-based approximations for the operator's update rule, especially in settings where exact DCOPF gradients are computationally expensive or unavailable in real time. Finally, although our current results focus on the single time-step case, extending the convergence and stability guarantees to multi-timestamp scenarios is an important direction for future work, especially for dynamic and time-coupled systems.

\bibliographystyle{IEEEtran}
\bibliography{main}
\end{document}